\documentclass[runningheads]{llncs}
\usepackage[T1]{fontenc}
\usepackage{amsmath, amssymb}
\usepackage{mathtools}
\usepackage{comment}

\usepackage{tikz}
\usetikzlibrary{fit,backgrounds, arrows.meta, calc, positioning, matrix}

\usepackage{algorithm}
\usepackage{algpseudocode}
\usepackage{placeins} % in preamble
\usepackage{url}
\newcommand{\doi}[1]{\url{https://doi.org/#1}}
\newcommand{\asc}{\operatorname{asc}}

\spnewtheorem{fact}[theorem]{Fact}{\bfseries}{\itshape}

\providecommand{\qedsymbol}{\ensuremath{\square}}
\newcommand{\qedbox}{\hfill\qedsymbol}

\spnewtheorem{claimnum}[theorem]{Claim}{\bfseries}{\itshape}



\title{Exact Accepting-State Spectrum for Reversal of Permutation Automata}
\titlerunning{}

\usepackage{bbding}
\author{Samuel German \Envelope}
\authorrunning{Samuel German}
\institute{University of California, San Diego, USA \\
\email{sgerman@ucsd.edu}}

\begin{document}
\maketitle
\begin{abstract} In this paper we determine the accepting-state spectrum of the reversal operation on permutation automata, and moreover prove the conjecture conjectured by Rauch and Holzer in so doing. Namely, we show that reversal satisfies the spectrum 

\[
g^{\asc}_{R,\mathrm{PFA}}(m)=
\begin{cases}
\{0\}, & \text{if } m=0,\\[1mm]
\{1\}, & \text{if } m=1,\\[1mm]
\mathbb{N}_{\ge 2}, & \text{if } m\ge 2.
\end{cases}
\]
We do this by, for
every $m\ge 2$ and every $\alpha\ge 2$, constructing a binary permutation
automaton $A_{m,\alpha}$ such that
$
\asc(L(A_{m,\alpha}))=m$ and $
\asc(L(A_{m,\alpha})^{R})=\alpha$. The spectrum above follows from this (combined with the trivial cases $m=0$ and $m=1$ and the fact that $1$ is magic for every $m\ge 2$).

Thus our result shows that for $m \ge 2$, the reversal operation in permutation automaton attains every single positive value in the accepting-state complexity except for $1$. 

We use some group theory to construct the automaton which actually witnesses the values we show to be in the accepting-state complexity. In particular, what we do is that we make the states of said automaton the subsets of size $\alpha$ of $[n]$, where we have set $n=m+\alpha-1$, under the action generated by an $n$-cycle and a
transposition, while the accepting states form a single star family. After applying the reversal operation to the language recognized by this automaton, we show that there is an automaton recognizing the language so produced for which its  reachable states are the stars, enabling us to count the reachable states which are also final states in a precise manner, and to prove that the reachable part of the reverse automaton is minimal.
 
\end{abstract}
\section{Introduction}
\label{sec:introduction}
We give an exact expression for the accepting-state spectrum of the reversal operation for permutation automata. Our result also constitutes a proof of a conjecture of Rauch and Holzer; in particular, they conjectured that once the trivial cases $m=0$ and $m=1$ are separated, the
only magic value for reversal is $1$~\cite[Conjecture~3.25]{RauchHolzer2023}, and this conjecture of theirs follows from our exact accepting-state expression. 

Accepting-state complexity is a concept which was introduced and studied by Dassow~\cite{Dassow2016}. It measures the number of accepting states in the minimal automata (automata with the fewest states) recognizing a language. 

Accepting-state complexity in the case of permutation automata underwent systematic study by Rauch and Holzer~\cite{RauchHolzer2023}. They focus on permutation automata due to its relevance in, for example, decomposing general automata, and study accepting-state complexity for permutation automata in particular as a way to better understand this measure, due to its divergence in behavior from other descriptional complexity measures such as state complexity~\cite{RauchHolzer2023}. For reversal  they left what happens to the accepting-state complexity open, and made a conjecture as to the magic values.

Some work Rauch and Holzer already did for this reversal question is that they proved that the value $1$ is magic for every
$m\ge2$, completely settled the case $m=2$, and conjectured that no further
obstruction exists~\cite{RauchHolzer2023}. We constructively prove this conjecture by giving the accepting-state spectrum exactly. What we have done is construct, for every $m\ge2$ and every
$\alpha\ge2$, a binary permutation automaton $A_{m,\alpha}$ satisfying 
$\asc(L(A_{m,\alpha}))=m$ and $\asc(L(A_{m,\alpha})^R)=\alpha$.
This means that 
\[
g^{\asc}_{R,\mathrm{PFA}}(m)=
\begin{cases}
\{0\}, & \text{if } m=0,\\[1mm]
\{1\}, & \text{if } m=1,\\[1mm]
\mathbb{N}_{\ge2}, & \text{if } m\ge2,
\end{cases}
\]
 thereby proving the aforementioned conjecture conjectured by Rauch and Holzer.

We employ some group theory in our proof. What we will see is that  for $m\ge2$ and
$\alpha\ge2$, the states of the automaton we construct will become the subsets of size $\alpha$ of $[n]$ for  where $n=m+\alpha-1$, under the natural action of the symmetric
group generated by an $n$-cycle and a transposition. The final states of the automaton we construct constitute the star
$\mathcal{S}(T_0)$ determined by an $(\alpha-1)$-subset $T_0$ of $[n]$.  After reversal, the reachable states become the stars
$\mathcal{S}(T)$, and the ones of these which accept are those stars for which
$T\subseteq q_I$, where $q_I$ is the initial state of the automaton we constructed at first. As a result, there are exactly $\alpha$ accepting states in the reachable
portion of the reverse automaton, and the same symmetric-group action
shows that this reachable portion is minimal. And we get a binary alphabet to witness the entire accepting state spectrum.

Furthermore, Theorem~\ref{thm:affine-witnesses} further shows that a binary
witness language $L$ can be chosen so that the minimal DFAs for $L$ and $L^R$
have metacyclic transition groups and the state counts are polynomial in $m$ and $\alpha$.

Conceptually, our accepting-state spectrum expression shows there is  only one nontrivial
obstruction on the accepting-state spectrum of permutation automata which is imposed by the reversal operation, namely $1$ in how it remains magic  for every $m\ge2$, but every larger value is part of the accepting-state spectrum. It is in this exact sense that the reversal operation can be said to have the simplest possible accepting-state spectrum which is compatible with the obstruction which Rauch and Holzer discovered~\cite{RauchHolzer2023}.

\section{Preliminaries}
\label{sec:preliminaries}

Herein we only consider deterministic, complete automata. Another one of our conventions is that we operate in accordance with the standard meaning $[n] := \{1,\dots,n\}$.

A \emph{deterministic finite automaton} (DFA) is a tuple
$A = (Q,\Sigma,\cdot,s,F)$, where $Q$ is a finite state set, $\Sigma$ is a
finite input alphabet, $\cdot : Q \times \Sigma \to Q$ is the transition
function, $s \in Q$ is the initial state, and $F \subseteq Q$ is the set of
final states. The language accepted by $A$ (we also say recognized by $A$) is
\[L(A) := \{\, w \in \Sigma^{*} \mid s \cdot w \in F \,\}.\]
We say that a DFA is \emph{minimal} if there doesn't exist any other DFA which recognizes the same language as it but which has fewer states. 

A \emph{permutation automaton} is a DFA in which, for every $a \in \Sigma$, the map
$q \mapsto q \cdot a$ is a permutation of $Q$. Hence every word
$w \in \Sigma^{*}$ induces a permutation $\pi_{w}$ of $Q$, given by
$q \pi_{w} = q \cdot w$ for $q \in Q$.

We extend the transition function $\cdot$ to words by setting $q\cdot\varepsilon=q$ and $q\cdot(wa)=(q\cdot w)\cdot a$ for $q\in Q$, $w\in\Sigma^*$, and $a\in\Sigma$. Since each letter acts as a permutation of $Q$, each word $w\in\Sigma^*$ induces a permutation $\pi_w$ of $Q$, given by $q\pi_w=q\cdot w$ for all $q\in Q$.

For a regular language $L$, the \emph{accepting-state complexity} of $L$,
denoted by $\asc(L)$, is the number of final states in the DFA accepting $L$ which has the smallest number of final states out of any DFA which accepts $L$. 

For a language $K \subseteq \Sigma^{*}$, we write
$
K^{R} := \{\, w^{R} \mid w \in K \,\}
$
for the reversal of $K$.

\[
g^{\asc}_{R,\mathrm{PFA}}(m)
  := \left\{\, \alpha \ge 0 \;\middle|\;
  \begin{aligned}[t]
    &\exists\,K \text{ accepted by a permutation automaton}\\
    &\text{with }\asc(K)=m \text{ and } \asc(K^{R})=\alpha
  \end{aligned}
  \right\}.
\]

The following standard fact, due to Dassow~\cite{Dassow2016}, allows us to
identify accepting-state complexity directly from a minimal witness automaton.

\begin{proposition}[{\cite{Dassow2016}}]
\label{prop:minimal-realizes-asc}
Let $A=(Q,\Sigma,\cdot,s,F)$ be a minimal DFA. Then
\[
\asc(L(A)) = |F|.
\]
\end{proposition}

\begin{lemma}
\label{lem:positive-words-generate-group}[{cf.~\cite[Chapter~II, Proposition~3.13]{Pin2025}}]
Let $X$ be a finite set and let $c_1,\dots,c_r\in \mathrm{Sym}(X)$. Then
the submonoid of $\mathrm{Sym}(X)$ generated by $c_1,\dots,c_r$ is equal
to the subgroup $\langle c_1,\dots,c_r\rangle$. Equivalently, every
element of $\langle c_1,\dots,c_r\rangle$ is induced by some positive word
over $c_1,\dots,c_r$.
\end{lemma}

\begin{lemma}[Reverse subset construction]
\label{lem:reverse-subset}
Let $A=(Q,\Sigma,\cdot,s,F)$ be some DFA. Define
$A^{R} = \bigl(2^{Q},\Sigma,\delta^{R},F,\mathcal{F}\bigr)$,
wherein it is
$\delta^{R}(S,a)
   := \{\, q \in Q \mid q \cdot a \in S \,\} $ for all
   $S \subseteq Q,\ a \in \Sigma$, and
$
\mathcal{F} := \{\, S \subseteq Q \mid s \in S \,\}$.
Then
$L(A^{R}) = L(A)^{R}$.
That is, for every $S\subseteq Q$ and every $w\in\Sigma^{*}$,
$\delta^{R}(S,w)=\{\, q\in Q \mid q\cdot w^{R}\in S \,\}$.
Moreover, if $A$ is a permutation automaton, then every letter induces a bijection through the transition function from $Q$ to $Q$, and thus we have here 
$
\delta^{R}(S,a) = a^{-1}(S) $ for all $
S \subseteq Q,\ a \in \Sigma$.
In particular, each letter induces a bijection when the transition function fixes it from $2^{Q}$ to the same set $2^{Q}$, so $A^{R}$ is itself
a permutation automaton. As such, we have that for every $w\in\Sigma^{*}$,
$\delta^{R}(S,w)=\pi_{w^{R}}^{-1}(S)$.
\end{lemma}

\begin{proof}

Our first objective will be to prove the word-level formula, the proving of which we tackle by induction on $|w|$. In the base case of
$w=\varepsilon$, then
$\delta^{R}(S,\varepsilon)=S=\{\, q\in Q \mid q\cdot \varepsilon^{R}\in S \,\}$.
Now for the inductive hypothesis let $w=ua$, where $u\in\Sigma^{*}$ and $a\in\Sigma$, and assume that the
formula holds for $u$. Then
\begin{align*}
\delta^{R}(S,ua)
 &= \delta^{R}(\delta^{R}(S,u),a) \\
 &= \{\, q\in Q \mid q\cdot a \in \delta^{R}(S,u) \,\} \\
 &= \{\, q\in Q \mid (q\cdot a)\cdot u^{R} \in S \,\} \\
 &= \{\, q\in Q \mid q\cdot (ua)^{R} \in S \,\}.
\end{align*}
This proves
$\delta^{R}(S,w)=\{\, q\in Q \mid q\cdot w^{R}\in S \,\} $ for all
$S\subseteq Q,\ w\in\Sigma^{*}$.

Hence, for every $w\in\Sigma^{*}$,
\begin{align*}
w\in L(A^{R})
 &\iff \delta^{R}(F,w)\in\mathcal{F} \\
 &\iff s\in \delta^{R}(F,w) \\
 &\iff s\cdot w^{R}\in F \\
 &\iff w^{R}\in L(A) \\
 &\iff w\in L(A)^{R}.
\end{align*}
Therefore $L(A^{R})=L(A)^{R}$.

If $A$ is a permutation automaton, then for each $a\in\Sigma$ the map
$q\mapsto q\cdot a$ is a bijection of $Q$. Hence
$\delta^{R}(S,a)=\{\, q\in Q \mid q\cdot a\in S \,\}=a^{-1}(S)$.
The inverse of the map $S\mapsto a^{-1}(S)$ on $2^{Q}$ is
$S\mapsto a(S)$, so each letter is a bijection from $2^{Q}$ to the selfsame set, and thus
$A^{R}$ is a permutation automaton. Finally, for every $w\in\Sigma^{*}$,
$\delta^{R}(S,w)=\{\, q\in Q \mid q\pi_{w^{R}}\in S \,\}=\pi_{w^{R}}^{-1}(S)$. \qedbox
\end{proof}

We also use standard group-action notation. For $n \ge 1$, let $S_{n}$
denote the symmetric group on $[n]$. If $G \le S_{n}$, then $G$ acts on
$[n]$ by $i\pi = \pi(i)$ for $i \in [n]$ and $\pi \in G$, and we write
\[iG := \{\, i\pi \mid \pi \in G \,\}\] for the corresponding orbit. We also
write \[\operatorname{Stab}_{G}(i) := \{\, \pi \in G \mid i\pi = i \,\}\] for the
point stabilizer of $i$. We treat the letters of the alphabet of the automata as functions, composing them.  Thus, if
$w=c_{1}c_{2}\cdots c_{r}$, then the permutation induced by $w$ is
$\pi_{w}=c_{r}\circ \cdots \circ c_{2}\circ c_{1}$, so that
$q\pi_{w}=q\cdot w$ for every state $q$. 

Additionally, for $X\subseteq[n]$ and $\pi\in S_n$, we write $X\pi$ for 
$\{\pi(i):i\in X\}$. The induced action of $S_n$ on the
$r$-element subsets of $[n]$ is transitive, since any bijection
between two such subsets extends to a permutation of $[n]$.

The following facts are standard, and we use them throughout: for $n\ge 2$, the cycle
$(1\,2\,\dots\,n)$ generates
$S_n$ in tandem with the transposition $(1\,2)$; the natural action of $S_n$ on $[n]$ is transitive; and
that the point stabilizer $\operatorname{Stab}_{S_k}(k)$ acts transitively
on $[k-1]$ (see, e.g.,
Dixon--Mortimer~\cite[Chapter~1]{DixonMortimer1996}). 

\section{Exact Spectrum of Reversal}
\label{sec:reversal}

We now determine the reversal spectrum exactly. The cases $m=0$ and $m=1$
will be handled at the end of the section, so throughout this section we fix
integers $m \ge 2$ and $\alpha \ge 2$, and set
$n := m+\alpha-1$.

Our witness family is based on the natural action of $S_n$ on the family of
$\alpha$-subsets of $[n]$. Let
$Q := \binom{[n]}{\alpha}$,
$q_I := \{1,\dots,\alpha\}$, and $
T_0 := \{1,\dots,\alpha-1\} $.
For each $T \in \binom{[n]}{\alpha-1}$, define
$\mathcal{S}(T) := \{\, X \in Q \mid T \subseteq X \,\}$.
In particular, $\mathcal{S}(T_0)$ is the family of all $\alpha$-subsets of
$[n]$ containing $T_0$.

Let the alphabet be $\Sigma = \{a,b\}$, where the letter $a$ acts on $[n]$
as the $n$-cycle
$a=(1\,2\,\dots\,n)$
and the letter $b$ acts on $[n]$ as the transposition
$
b=(1\,2)$.
These actions induce permutations of $Q$ by
$
X \cdot a := \{\, ia \mid i \in X \,\}$, $
X \cdot b := \{\, ib \mid i \in X \,\} $ for all $
X \in Q$.

Finally, let
$
F := \mathcal{S}(T_0)$,
and define
$
A_{m,\alpha} := (Q,\Sigma,\cdot,q_I,F)$.
Since $a$ and $b$ generate $S_n$~\cite[Chapter~1]{DixonMortimer1996},
the automaton $A_{m,\alpha}$ is a binary permutation automaton.

\begin{lemma}
\label{lem:reversal-forward-minimal}
The automaton $A_{m,\alpha}$ is minimal and has exactly $m$ final states.
In particular,
$\asc(L(A_{m,\alpha})) = m$.
\end{lemma}

\begin{proof}
First, $F$ has exactly $m$ elements. Indeed, an $\alpha$-subset $X$ of
$[n]$ belongs to $F=\mathcal{S}(T_0)$ if and only if
$X = T_0 \cup \{u\}
$
for some $u \in [n]\setminus T_0$. Hence
$
|F| = |[n]\setminus T_0|
     = n-(\alpha-1)
     = m$

We next show that every state of $A_{m,\alpha}$ is reachable. By
Lemma~\ref{lem:positive-words-generate-group}, every permutation in
$\langle a,b\rangle=S_n$ is induced by some word in $\{a,b\}^{*}$. That for
every $X \in Q$ there exists a word $w \in \{a,b\}^{*}$ such that
$q_I \cdot w = X$ is a consequence of 
the natural action of $S_n$ on $\binom{[n]}{\alpha}$ being transitive.
Thus every state is reachable from $q_I$.

We still need to show that distinct states are distinguishable, which we now turn attention towards. Let
$X,Y \in Q$ with $X \neq Y$. Choose some element
$
x \in X \setminus Y$.
Since $|X|=\alpha$, there exists an $(\alpha-1)$-subset
$W \subseteq X$ such that $x \in W$. Choose any element
$
u \in [n]\setminus T_0$.
Because $|W|=|T_0|=\alpha-1$ and $|X\setminus W|=1$, we may choose a
bijection $W \to T_0$, send the unique element of $X\setminus W$ to $u$,
and extend this map to a permutation $\pi \in S_n$. Then
$
W\pi = T_0$ and 
$
(X\setminus W)\pi = \{u\}$.

By Lemma~\ref{lem:positive-words-generate-group}, there exists a word
$w \in \{a,b\}^{*}$ inducing $\pi$. Therefore
\[
X \cdot w = X\pi = T_0 \cup \{u\} \in F.
\]

On the other hand, because $x \in W\setminus Y$, we have $W \nsubseteq Y$.
Moreover, $W\pi=T_0$ and $\pi$ is bijective, so $\pi^{-1}(T_0)=W$. If
$T_0 \subseteq Y\pi$, then applying $\pi^{-1}$ would give $W \subseteq Y$,
a contradiction. Hence
$
T_0 \nsubseteq Y\pi = Y \cdot w$
and thus
$
Y \cdot w \notin F$.
So $X$ and $Y$ are distinguishable.

We have shown that all states are reachable and pairwise distinguishable, so
$A_{m,\alpha}$ is minimal. Proposition~\ref{prop:minimal-realizes-asc} now
yields
$
\asc(L(A_{m,\alpha})) = |F| = m$. \qedbox
\end{proof}

\begin{lemma}
\label{lem:reversal-reachable-stars}
The reachable states of the reverse automaton $A_{m,\alpha}^{R}$ are exactly
the stars
\[
\mathcal{S}(T)
\qquad
\left(T \in \binom{[n]}{\alpha-1}\right).
\]
More precisely, for every $T \in \binom{[n]}{\alpha-1}$ and every
$c \in \{a,b\}$,
$\delta^{R}(\mathcal{S}(T),c)=\mathcal{S}(Tc^{-1})$.
\end{lemma}

\begin{proof}
By Lemma~\ref{lem:reverse-subset}, the reverse automaton has the form
$A_{m,\alpha}^{R} = \bigl(2^{Q},\Sigma,\delta^{R},F,\mathcal{F}\bigr)$,
with initial state
$F=\mathcal{S}(T_0)$.

Let $T \in \binom{[n]}{\alpha-1}$ and let $c \in \{a,b\}$. For
$X \in Q$, we have
\begin{align*}
X \in \delta^{R}(\mathcal{S}(T),c)
&\iff X \cdot c \in \mathcal{S}(T) \\
&\iff T \subseteq X \cdot c \\
&\iff Tc^{-1} \subseteq X \\
&\iff X \in \mathcal{S}(Tc^{-1}).
\end{align*}
Hence
$\delta^{R}(\mathcal{S}(T),c)=\mathcal{S}(Tc^{-1})$,
as claimed.

Thus we have  that every state which is a reachable one from $F=\mathcal{S}(T_0)$ is necessitated to be a
star, of the form $\mathcal{S}(T)$, which is because that in the automaton produced by the reverse construction sends a star to a star. Thus every reachable state of
$A_{m,\alpha}^{R}$ belongs to the set
$\left\{\, \mathcal{S}(T) \mid T \in \binom{[n]}{\alpha-1} \,\right\}$.

Conversely, because the letter $a$ acts on stars by $a^{-1}$ and the letter
$b$ acts on stars by $b^{-1}=b$, the transformations of the centers
generated by the reverse automaton form the group
$\langle a^{-1}, b \rangle = \langle a,b \rangle = S_n$.
Now it is by Lemma~\ref{lem:positive-words-generate-group} that every element of this group
is induced on the centers by some word in $\{a,b\}^{*}$. That the natural
action of $S_n$ on $\binom{[n]}{\alpha-1}$ is transitive implies that for every
$T \in \binom{[n]}{\alpha-1}$ there exists a word $w \in \{a,b\}^{*}$
such that reading $w$ in the reverse automaton sends
$\mathcal{S}(T_0)$ to $\mathcal{S}(T)$, whence it follows that every such star is a reachable state. Therefore the reachable states of $A_{m,\alpha}^{R}$ are exactly the stars
\[
\mathcal{S}(T)
\qquad
\left(T \in \binom{[n]}{\alpha-1}\right).
\]

Finally, these stars are pairwise distinct. Indeed, if
$T,U \in \binom{[n]}{\alpha-1}$ with $T \neq U$, choose
$x \in T \setminus U$. Since
$
|[n]\setminus U| = n-(\alpha-1)=m \ge 2$,
there exists some $y \in [n]\setminus (U \cup \{x\})$. Then
$
U \cup \{y\} \in \mathcal{S}(U) $
but $
U \cup \{y\} \notin \mathcal{S}(T)$,
because $x \in T$ but $x \notin U \cup \{y\}$. Thus
$\mathcal{S}(T) \neq \mathcal{S}(U)$. \qedbox
\end{proof}

\begin{lemma}
\label{lem:reversal-accepting-stars}
Among the reachable states of $A_{m,\alpha}^{R}$, the accepting states are
exactly the stars
\[
\mathcal{S}(T)
\qquad
\left(T \in \binom{[n]}{\alpha-1},\ T \subseteq q_I\right).
\]
Consequently, the reachable part of $A_{m,\alpha}^{R}$ has exactly
$
\binom{\alpha}{\alpha-1}=\alpha
$
accepting states.
\end{lemma}

\begin{proof}
By Lemma~\ref{lem:reverse-subset}, a state $S \subseteq Q$ of the reverse
automaton $A_{m,\alpha}^{R}$ is accepting if and only if
$q_I \in S$.
By Lemma~\ref{lem:reversal-reachable-stars}, every reachable state of
$A_{m,\alpha}^{R}$ has the form $\mathcal{S}(T)$ for some
$T \in \binom{[n]}{\alpha-1}$. Therefore a reachable state
$\mathcal{S}(T)$ is accepting if and only if
$q_I \in \mathcal{S}(T)$.
By the definition of $\mathcal{S}(T)$, this is equivalent to
$T \subseteq q_I$.
Hence the accepting reachable states are exactly the stars
\[
\mathcal{S}(T)
\qquad
\left(T \in \binom{[n]}{\alpha-1},\ T \subseteq q_I\right).
\]

Since $|q_I|=\alpha$, the number of $(\alpha-1)$-subsets of $q_I$ is
$\binom{\alpha}{\alpha-1}=\alpha$.
Moreover, these stars are pairwise distinct by
Lemma~\ref{lem:reversal-reachable-stars}. Therefore the reachable part of
$A_{m,\alpha}^{R}$ has exactly $\alpha$ accepting states. \qedbox
\end{proof}

\begin{lemma}
\label{lem:reversal-reachable-minimal}
The reachable part of $A_{m,\alpha}^{R}$ is minimal.
\end{lemma}

\begin{proof}
By Lemma~\ref{lem:reversal-reachable-stars}, the reachable states of
$A_{m,\alpha}^{R}$ are exactly the stars
\[
\mathcal{S}(T)
\qquad
\left(T \in \binom{[n]}{\alpha-1}\right).
\]
Thus every state of the reachable part is reachable, and it remains only to
show that any two distinct reachable states are distinguishable.

Let $\mathcal{S}(S)$ and $\mathcal{S}(T)$ be distinct reachable states.
Then
\[
S,T \in \binom{[n]}{\alpha-1} \] and 
$
S \neq T$.

Choose an element
$
t \in T \setminus S$.
Since $m \ge 2$, we have
$n = m+\alpha-1 \ge \alpha+1$,
and therefore
$
\alpha+1 \notin q_I = \{1,\dots,\alpha\}$.

Because $t \notin S$ and $|S|=\alpha-1$, there exists a permutation
$\pi \in S_n$ such that
$S\pi = T_0 = \{1,\dots,\alpha-1\} $ and $
t\pi = \alpha+1$: to see this, consider that we can choose some bijection from $S$ onto $T_0$, map $t$ to
$\alpha+1$ and extend this to some permutation of $[n]$.

Now by Lemma~\ref{lem:reversal-reachable-stars}, the letter $a$ acts on star
centers by $a^{-1}$, and the letter $b$ acts on star centers by
$b^{-1}=b$. Hence the transformations of the centers, getting the transformations from what is induced by words in the
reverse automaton, form the submonoid of $S_n$ generated by $a^{-1}$ and
$b$. Since
$\langle a^{-1},b\rangle=\langle a,b\rangle=S_n$,
Lemma~\ref{lem:positive-words-generate-group} yields a word
$w\in\{a,b\}^{*}$ whose action on the centers is exactly $\pi$; it follows herefrom that
for every $U\in\binom{[n]}{\alpha-1}$,
$\delta^{R}(\mathcal{S}(U),w)=\mathcal{S}(U\pi)$.

Applying this with $U=S$, we obtain
$\delta^{R}(\mathcal{S}(S),w)=\mathcal{S}(S\pi)=\mathcal{S}(T_0)$.
Since $T_0 \subseteq q_I$, Lemma~\ref{lem:reversal-accepting-stars} shows
that $\mathcal{S}(T_0)$ is accepting.

On the other hand, $t \in T$ implies
$\alpha+1 = t\pi \in T\pi$.
Because $\alpha+1 \notin q_I$, it follows that
$T\pi \nsubseteq q_I$.
Hence Lemma~\ref{lem:reversal-accepting-stars} shows that
$\delta^{R}(\mathcal{S}(T),w)=\mathcal{S}(T\pi)
$
is nonaccepting.

Thus the word $w$ distinguishes $\mathcal{S}(S)$ and $\mathcal{S}(T)$, implying 
all reachable states of $A_{m,\alpha}^{R}$ are pairwise
distinguishable, and the reachable part of $A_{m,\alpha}^{R}$ is minimal. \qedbox
\end{proof}

\begin{theorem}
\label{thm:reversal-attainability}
For every $m \ge 2$ and every $\alpha \ge 2$, we have
$
\alpha \in g^{\asc}_{R,\mathrm{PFA}}(m)$.
Equivalently, the language $L(A_{m,\alpha})$ satisfies
\[
\asc(L(A_{m,\alpha})) = m
\qquad\text{and}\qquad
\asc(L(A_{m,\alpha})^{R}) = \alpha.
\]
\end{theorem}

\begin{proof}
By Lemma~\ref{lem:reversal-forward-minimal}, the automaton $A_{m,\alpha}$ is
minimal with exactly $m$ final states. Hence
$\asc(L(A_{m,\alpha})) = m
$
by Proposition~\ref{prop:minimal-realizes-asc}.

By Lemma~\ref{lem:reverse-subset}, the reverse subset construction
$A_{m,\alpha}^{R}$ accepts the language $L(A_{m,\alpha})^{R}$. By
Lemmata~\ref{lem:reversal-reachable-stars},
\ref{lem:reversal-accepting-stars}, and
\ref{lem:reversal-reachable-minimal}, the reachable part of
$A_{m,\alpha}^{R}$ is a minimal DFA with exactly $\alpha$ final states.
Since removing unreachable states does not change the accepted language,
Proposition~\ref{prop:minimal-realizes-asc} yields
$
\asc(L(A_{m,\alpha})^{R}) = \alpha$: so $\alpha \in g^{\asc}_{R,\mathrm{PFA}}(m)$ holds. \qedbox 
\end{proof}

\begin{corollary}
\label{cor:reversal-spectrum}
The accepting-state spectrum of the reversal operation for permutation automata is:
\[
g^{\asc}_{R,\mathrm{PFA}}(m)=
\begin{cases}
\{0\}, & \text{if } m=0,\\[1mm]
\{1\}, & \text{if } m=1,\\[1mm]
\mathbb{N}_{\ge 2}, & \text{if } m\ge 2.
\end{cases}
\]
In particular, this proves the conjecture of Rauch and
Holzer~\cite[Conjecture~3.25]{RauchHolzer2023}.
\end{corollary}

\begin{proof}
If $m=0$, then the only language $K$ with $\asc(K)=0$ is the empty
language. Hence $K^{R}=\emptyset$, and so
$g^{\asc}_{R,\mathrm{PFA}}(0)=\{0\}$.

If $m=1$, then the unary language $a^{*}$ is accepted by the one-state
all-final permutation automaton, so
$1\in g^{\asc}_{R,\mathrm{PFA}}(1)$.
Conversely, if $K$ is accepted by a permutation automaton and
$\asc(K)=1$, then $\asc(K^{R})=1$ by
\cite[Theorem~3.24]{RauchHolzer2023}. Therefore
$g^{\asc}_{R,\mathrm{PFA}}(1)=\{1\}$.

Now let $m \ge 2$. By Theorem~\ref{thm:reversal-attainability}, every
$\alpha \ge 2$ belongs to $g^{\asc}_{R,\mathrm{PFA}}(m)$. On the other
hand, $\alpha=0$ is impossible because $\asc(K)=m\ge 2$ implies
$K\neq \emptyset$, and therefore $K^{R}\neq \emptyset$. Also $\alpha=1$
is impossible by \cite[Lemma~3.23]{RauchHolzer2023}. Therefore
\[
g^{\asc}_{R,\mathrm{PFA}}(m)=\mathbb{N}_{\ge 2}
\qquad (m\ge 2).
\]
This is exactly the statement conjectured in
\cite[Conjecture~3.25]{RauchHolzer2023}. \qedbox
\end{proof}

\begin{theorem}
\label{thm:affine-witnesses}
Let $m,\alpha\ge 2$, let $\ell=\operatorname{lcm}(m,\alpha)$,
and let $p\equiv 1\pmod{\ell}$ be prime. There is a binary
permutation language $L$ with
$\operatorname{asc}(L)=m$ and $\operatorname{asc}(L^R)=\alpha$.
The minimal DFAs for $L$ and $L^R$ have $p\ell/\alpha$ and
$p\ell/m$ states, respectively, and both have metacyclic transition
group $C_p\rtimes C_\ell$.
\end{theorem}

\begin{proof}
Let $u\in\mathbb F_p^\times$ be such that $u$ is of order $\ell$. 
$a(x)=ux$ and $b(x)=x+1$ generate the affine group $G$ of maps
$x\mapsto tx+s$ with $t\in\langle u\rangle$ and
$s\in\mathbb F_p$. Thus $G\cong C_p\rtimes C_\ell$.
Let $H$ be the subgroup of order $\alpha$ in the stabilizer
of $0$, and let $J$ be the subgroup of order $m$ in the
stabilizer of $1$.

Take $G/H$ to be the state set in the automata we construct to recognize $L$, which automata we refer to as being the ``forward automaton" in some places. Let $H$ be the initial state of this automaton. Have the letters act according to the left multiplication. Finally, take the set of final states to be $F=\{jH:j\in J\}$.

Now set $D=\{hJ:h\in H\}\subseteq G/J$, rendering $|F|=m$ and
$|D|=\alpha$ because $H\cap J=\{1_G\}$ (only the identity fixes
both $0$ and $1$).

It is the case that any $g$ stabilizing $F$ preserves the projected orbit $J(0)$
under $xH\mapsto x(0)$, thereby fixing the mean $1$ of that orbit
and sends $0$ to $j(0)$ for some $j\in J$, forcing $g=j$
by agreement at $0$ and $1$ and giving
$\operatorname{Stab}_G(F)=J$, with the symmetric argument via
$xJ\mapsto x(1)$ giving $\operatorname{Stab}_G(D)=H$.

We now argue that the forward automata which we defined in order to recognize $L$, as well as the reachable reverse
automata, are both minimal. First of all, by the reverse subset construction, the reachable reverse states
are the translates $gF$, indexed by $gJ$. Now we know that a forward state
$xH$ lies in $gF$ exactly when $gJ\in xD$. Since $D$ has
stabilizer $H$, distinct forward states differ in membership
in some translate $gF$, and a word inducing $g^{-1}$
distinguishes them. Finsally, distinct reverse subsets are distinguished
by reversing a word that reaches a state in their symmetric
difference from $H$, giving what we set out to prove in this paragraph. 

It is that exactly $|D|=\alpha$ of the reverse subsets contain $H$.
The states are $[G:H]$ and $[G:J]$, in number. Both coset actions
are faithful because their projections onto the affine action
on $\mathbb F_p$ are surjective; they commute with the action of $G$.
\end{proof}

By Linnik's theorem~\cite{Xylouris2009}, $p$ can be chosen
polynomially bounded in $\ell$. Since $\ell\le m\alpha$,
both state counts are polynomially bounded in $m$ and $\alpha$.

\emph{Example.} We  here record a small example
with $m=3$, in light of how the case $m=2$ was already settled by Rauch and
Holzer~\cite[Theorem~3.24]{RauchHolzer2023}. Consider the witness $A_{3,4}$. Here
$
n = m+\alpha-1 = 3+4-1 = 6$.
Write $1234$ for $\{1,2,3,4\}$, and similarly for other subsets of
$[6]$. Then
\[
q_I=1234,
\qquad
T_0=123,
\qquad
F=\mathcal{S}(123)=\{1234,1235,1236\}.
\]
It is thus the case that $A_{3,4}$ has exactly $m=3$ final states.

In the reverse automaton $A_{3,4}^{R}$, the initial state is
$\mathcal{S}(123)$. By Lemma~\ref{lem:reversal-reachable-stars}, reading a letter
$c\in\{a,b\}$ starting from $S(T)$ and following where the transition functions says to go based on the letter and state gives the star centered
at $\{i c^{-1}:i\in T\}$. Thus what $a$ does it move each element of
$T$ one step backward around the cycle $(1\,2\,3\,4\,5\,6)$,
whereas $b$ swaps $1$ and $2$ due to how $b^{-1}=b$, so what ensues as a result is:
\[
\delta^{R}(\mathcal{S}(T),a)=\mathcal{S}(Ta^{-1}),
\qquad
\delta^{R}(\mathcal{S}(T),b)=\mathcal{S}(Tb).
\]
Starting from $\mathcal{S}(123)$, reading repeated copies of $a$ produce
\[
\mathcal{S}(123)\xrightarrow{a}\mathcal{S}(126)\xrightarrow{a}\mathcal{S}(156)
\xrightarrow{a}\mathcal{S}(456)\xrightarrow{a}\mathcal{S}(345)
\xrightarrow{a}\mathcal{S}(234).
\]

As such $\mathcal{S}(234)$ is reachable. Reading $b$ but one more time gives
\[
\mathcal{S}(234)\xrightarrow{b}\mathcal{S}(134),
\]
and directly computing gives
\[
\delta^{R}(\mathcal{S}(123),a^{2}ba^{4})=\mathcal{S}(124).
\]

This implies that all of the four stars having centers being size $3$ subsets of $q_I=1234$, shown below, are reachable:

So all four stars whose centers are $3$-subsets of $q_I=1234$ are
explicitly reachable:
\[
\mathcal{S}(123),\quad
\mathcal{S}(124),\quad
\mathcal{S}(134),\quad
\mathcal{S}(234).
\]
The stars here are
\[
\begin{aligned}
\mathcal{S}(123)&=\{1234,1235,1236\}, &
\mathcal{S}(124)&=\{1234,1245,1246\}, \\
\mathcal{S}(134)&=\{1234,1345,1346\}, &
\mathcal{S}(234)&=\{1234,2345,2346\}.
\end{aligned}
\]
By Lemma~\ref{lem:reversal-accepting-stars}, a reachable star
$\mathcal{S}(T)$ is accepting if and only if $T \subseteq q_I$. Hence
these four explicitly listed stars are exactly the accepting states of the
reachable part of $A_{3,4}^{R}$. Therefore that reachable part has exactly
$\binom{4}{3}=4
$
accepting states. In particular,
$
\asc(L(A_{3,4}))=3 $ and 
$
\asc(L(A_{3,4})^{R})=4$.
This gives a concrete illustration of
Theorem~\ref{thm:reversal-attainability} and
Corollary~\ref{cor:reversal-spectrum} in a case $m=3$ beyond the previously
settled case $m=2$.

\bibliographystyle{splncs04}
\bibliography{references}

\end{document}